\documentclass[10pt,final,conference,twocolumn]{IEEEtran}

\usepackage[cmex10]{amsmath}		
\usepackage{cite}
\usepackage{graphicx,color,epsfig,rotating}
\usepackage{amsfonts,amsmath,amssymb}
\usepackage{subfigure}
\usepackage{url}
\usepackage{epsfig}
\usepackage{multirow}
\newtheorem{definition}{Definition}
\newtheorem{theorem}{Theorem}
\newtheorem{lemma}{Lemma}

\newtheorem{proof}{Proof}

\makeatother
\begin{document}
\title{\LARGE{Scaling Behaviors  of Wireless Device-to-Device Communications with Distributed Caching}}
\author{Negin Golrezaei,~Alexandros G. Dimakis,~Andreas F. Molisch
\\ Dept. of Electrical Eng.\\
University of Southern California\\
emails:~\{golrezae,dimakis,molisch\}@usc.edu
}
\maketitle

\begin{abstract}
We analyze a novel architecture for caching popular video content to enable wireless device-to-device collaboration. We focus on the asymptotic scaling characteristics 
and show how they depends on video content popularity statistics. We identify a fundamental conflict between collaboration distance and interference and show how to optimize the transmission power to maximize frequency reuse. 

 Our main result is a closed form expression of the optimal collaboration distance as a function of the model parameters. Under the common assumption of a Zipf distribution for content reuse, we show that if the Zipf exponent is greater than 1, it is possible to have a number of D2D interference-free collaboration pairs that scales linearly in the number of nodes. If the Zipf exponent is smaller than 1, we identify the best possible scaling in the number of D2D collaborating links. Surprisingly, a very simple distributed caching policy achieves the optimal scaling behavior and therefore there is no need to centrally coordinate what each node is caching.	
	
\end{abstract}
\section{Introduction}

Wireless mobile data traffic is expected to increase by a factor of $40$ over the next five years, from the current $93$ Petabytes to $3600$ Petabytes per month in the next five years~\cite{cisco66}. This explosive demand is fueled mainly by mobile video traffic that is expected to increase by a factor of $65$, and become the by far dominant source of data traffic. Since the available spectrum is physically limited and the spectral efficiency of current systems is already close to optimum, the main method for meeting this increased demand is to bring content closer to the users. Femto base stations \cite{chandrasekhar2008femtocell} are currently receiving a lot of attention for this purpose. 

A significant bottleneck in such small-cell architectures is that each station requires a high-rate backhaul link. Helper stations that replace high-rate backhaul with storage ~\cite{femtocaching}\cite{coded_femtocaching}, can ameliorate the problem, but still require additional infrastructure and have limited flexibility. 

To circumvent these problems, we recently proposed the use of device-to-device (D2D) communications combined with video caching in mobile devices~\cite{ICC_Workshop}  \cite{TWC}. The approach is based on three key observations: (i) Modern smartphones and tablets have significant storage capacity, (ii) video has a large amount of \textit{content reuse}, i.e., a small number of video files accounts for a large fraction of the traffic. 
(iii) D2D communication can occur over very short distances thus allowing high frequency reuse. Our proposed architecture functions as follows: users can collaborate by caching popular content and utilizing local D2D communication when a user in the vicinity requests a popular file. The base station can keep track of the availability of the cached content and direct requests to the most suitable nearby device; if there is no suitable nearby device, the BS supplies the requested video file directly, via a traditional downlink transmission. Storage allows users to collaborate even when they do not request the same content \textit{at the same time}. This is a new dimension in wireless collaboration architectures beyond relaying and cooperative communications as in~\cite{magazine} \cite{ICC_Workshop} and references therein.
%
%
%

A D2D video network can be analyzed using a protocol model, which means that only two devices that are within a "collaboration distance" of each other can exchange video files, while devices with a larger distance do not create any useful signal, but also no interference, for each other. The choice of the collaboration distance represents a tradeoff between two counteracting effects: decreasing the collaboration distance increases the frequency reuse and thus the potential throughput, but on the other hand decreases the probability that a device can find a requested file cached on another device within the collaboration distance. In \cite{TWC} we described this tradeoff and provided numerical solutions for the optimum distance, and the resulting system throughput. 

In the current paper we concentrate on the analytical treatment of the {\em scaling behavior} of a D2D network, i.e., how the throughput scales as the number of nodes increases. For conventional ad-hoc networks, scaling behavior has been derived in the seminal paper by Gupta and Kumar \cite{gupta2000capacity} has further 
received significant attention (\textit{e.g.} see \cite{Tse,grossglauser2001mobility,franceschetti2009capacity}).
%
This architecture not only differs from ad-hoc or collaborative networks in its application, but also shows a fundamentally different behavior due to its dependence on the video reuse statistics. We provide a closed form expression of the optimal collaboration distance as a function of the content reuse distribution parameters.

We model the request statistics for video files by a Zipf distribution which has been shown to fit well with measured YouTube video requests~\cite{tracedata} \cite{zipf}. We find that the scaling laws depend critically on the Zipf parameter, \textit{i.e.}, on the concentration of the request distribution. We show that if the Zipf exponent of the content reuse distribution is greater than $1$, it is possible 
to have a number of D2D interference-free collaboration pairs that scales linearly with the number of nodes. If the Zipf exponent is smaller than $1$, we identify the best possible scaling in the number of D2D collaborating links.  
 Surprisingly, a very simple distributed caching policy achieves the optimal scaling behavior and therefore there is no need to centrally coordinate what each node is caching. For Zipf exponent equal to $1$, we find the best collaboration distance and the best possible scaling. 

The remainder of this paper is organized as follows: In Section~\ref{sec:model} we set up the D2D formulation 
and explain the tradeoff between collaboration distance and interference. 
Section~\ref{sec:analysis} contains our two main theorems, the scaling behavior for Zipf exponents 
greater, smaller than and equal to $1$. In Section \ref{sec:discussion} we discuss future directions, open problems and conclusions. Finally, the Appendix contain the proofs of our theorems.

\section{Model and Setup}

\label{sec:model}
In this section, we discuss the fundamental system model; for a discussion of the assumptions, and justifications of simplifications, we refer the interested reader to \cite{TWC}. 

Assume a cellular network where each cell/base station (BS) serves $n$ users. For simplicity we assume that the cells are square, and we neglect inter-cell interference, so that we can consider one cell in isolation. Users are distributed randomly and independently in the cell. We assume that the D2D communication does not interfere with
the base station that can serve video requests that cannot be otherwise covered. For that reason, our only concern is the maximization of the number of D2D collaboration links that can be simultaneously scheduled. We henceforth do not need to consider explicitly the BS and its associated communications. 


The communication is modeled by a standard protocol model on a random geometric graph (RGG) $G(n,r(n))$. In this model users are randomly and uniformly distributed in a square (cell) of size $1$. Two users (assuming D2D communication is possible) can communicate if their euclidean distance is smaller than some collaboration distance $r(n)$~\cite{gupta2000capacity,RGG book}. The maximum allowable distance for D2D communication $r(n)$ is determined by the power level for each transmission. Figure \ref{RGG} illustrates an example of an RGG.
\begin{figure}
\centerline{\includegraphics[width=5.5 cm]{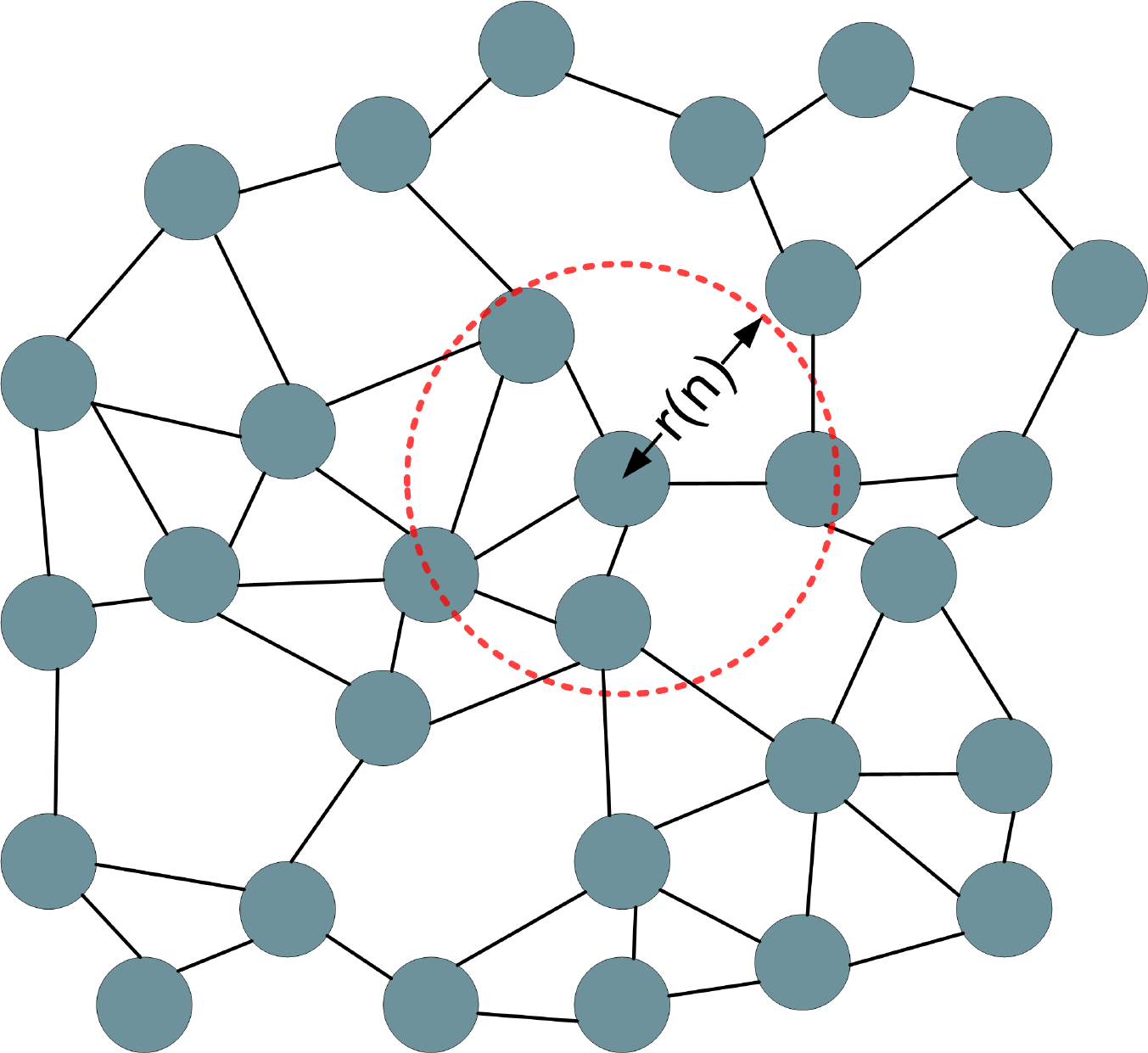}}
\caption{Random geometric graph example with collaboration distance $r(n)$.}
\label{RGG}
\end{figure}

We assume that users may request files from a set of size $m$ that we call a ``library''.  
The size of this set should increase as a function of the number of users $n$. Intuitively, the set of YouTube videos requested in Berkeley in one day should be smaller than the set of requested in Los Angeles. 
We assume that this growth should be sublinear in $n$, \textit{e.g.} $m$ could be $\Theta(\log(n))$ \footnote{We use the standard Landau notation: $f(n)=O(g(n))$ and $f(n)=\Omega(g(n))$ respectively denote $|f(n)|\leq c_1 g(n)$ and $|f(n)|\geq c_2 g(n)$ for some constants $c_1,c_2$. $f(n)=\Theta(g(n))$, stands for $f(n)=O(g(n))$ and $f(n)=\Omega(g(n))$.  Little-o notation, \textit{i.e.,} $f(n) = o(g(n))$ is equivalent to $\lim_{n\rightarrow \infty} \frac{f(n)}{g(n)}=0$.}.

Each user requests a file from the library by sampling independently using a popularity distribution. 
Based on several studies, Zipf distributions have been established as good models for the measured popularity of video files \cite{zipf,tracedata}. Under this model, the popularity of the $i$th popular file, denoted by $f_i$, is inversely proportional to its rank:
 \begin{equation}\label{zipf}
 f_i=\frac{{\frac{1}{{{i^{\gamma_r} }}}}}{{\sum\limits_{j = 1}^m {\frac{1}{{{j^{\gamma_r} }}}} }},\,\,\ 1\leq i\leq m.
 \end{equation}
 The Zipf exponent $\gamma_r$ characterizes the distribution by controlling the relative popularity of files.
 Larger $\gamma_r$ exponents correspond to higher content reuse, \textit{i.e.,} the first few popular files account for the majority of requests.

Each user has a storage capacity called cache which is populated with some video files. For our scaling law analysis we assume that all files have the same size, and each user can store one file. This yields a clean formulation and can be easily extended for larger storage capacities. 
 
Our scheme works as follows: If a user requests one of the files stored in neighbors' caches in the RGG, neighbors will handle the request locally through D2D communication; otherwise, the BS should serve the request. 
 Thus, to have D2D communication it is not sufficient that the distance between two users be less than $r(n)$; users should find their desired files locally in caches of their neighbors. A link between 
two users will be called potentially \textit{active} if one requests a file that the other is caching. 
 Therefore, the probability of D2D collaboration opportunities depends on what is stored and requested by the users. 
 
The decision of what to store can be taken in a distributed or centralized way.  
A central control of the caching by the BS allows very efficient file-assignment to the users. However, if such control is not desired or the users are highly mobile, caching has to be optimized in a distributed way. The simple randomized caching policy we investigate makes each user choose which file to cache by sampling from a caching distribution. 
It is clear that popular files should be stored with a higher probability, but the question is how much redundancy we want to have in our distributed cache. 


 

We assume that all D2D links share the same time-frequency transmission resource within one cell area. This is possible since the distance between requesting user and user with the stored file will typically small. However, there should be no interference of a transmission by others on an active D2D link. We assume that (given that node $u$ wants to transmit to node $v$) any transmission within range $r(n)$ from $v$ (the receiver) can introduce interference for the $u-v$ transmission. Thus, they cannot be activated simultaneously. 
This model is known as \emph{protocol model}; while it neglects important wireless propagation effects such as fading \cite{Molisch_book_2011}, it can provide fundamental insights and has been widely used in prior literature \cite{gupta2000capacity}.

To model interference given a storage configuration and user requests we start with all 
potential D2D collaboration links. Then, we construct the conflict graph as follows. We model any possible D2D link  between node $u$ as transmitter to node $v$ as a receiver with a vertex $u-v$ in the conflict graph. Then, we draw an edge between any two vertices (links) that create interference for each other according to the protocol model. Figure \ref{fig:subfig3} shows how the RGG in Figure \ref{fig:subfig1} is converted to the conflict graph. In Figure \ref{fig:subfig1}, receiver nodes are green and transmitter nodes are yellow. The nodes that should receive their desired files from the BS are gray. 
A set of D2D links is called active if they are potentially active and can be scheduled simultaneously, 
\textit{i.e.}, form an independent set in the conflict graph. The random variable counting the number of active D2D links under some policy is denoted by $L$. 

Figure~\ref{fig:subfig3} shows the conflict graph and one of maximum independent sets for the conflict graph. We can see that out of $14$ possible D2D links $9$ links can co-exist without interference. 
As is well known, determining the maximum independent set of an arbitrary graph is computationally intractable (NP complete~\cite{lawler1980generating}). 
Despite the difficulty of characterizing the number of interference-free active links, we can determine the best possible scaling law in our random ensemble.

\begin{figure}
\centerline{\includegraphics[scale=0.45]{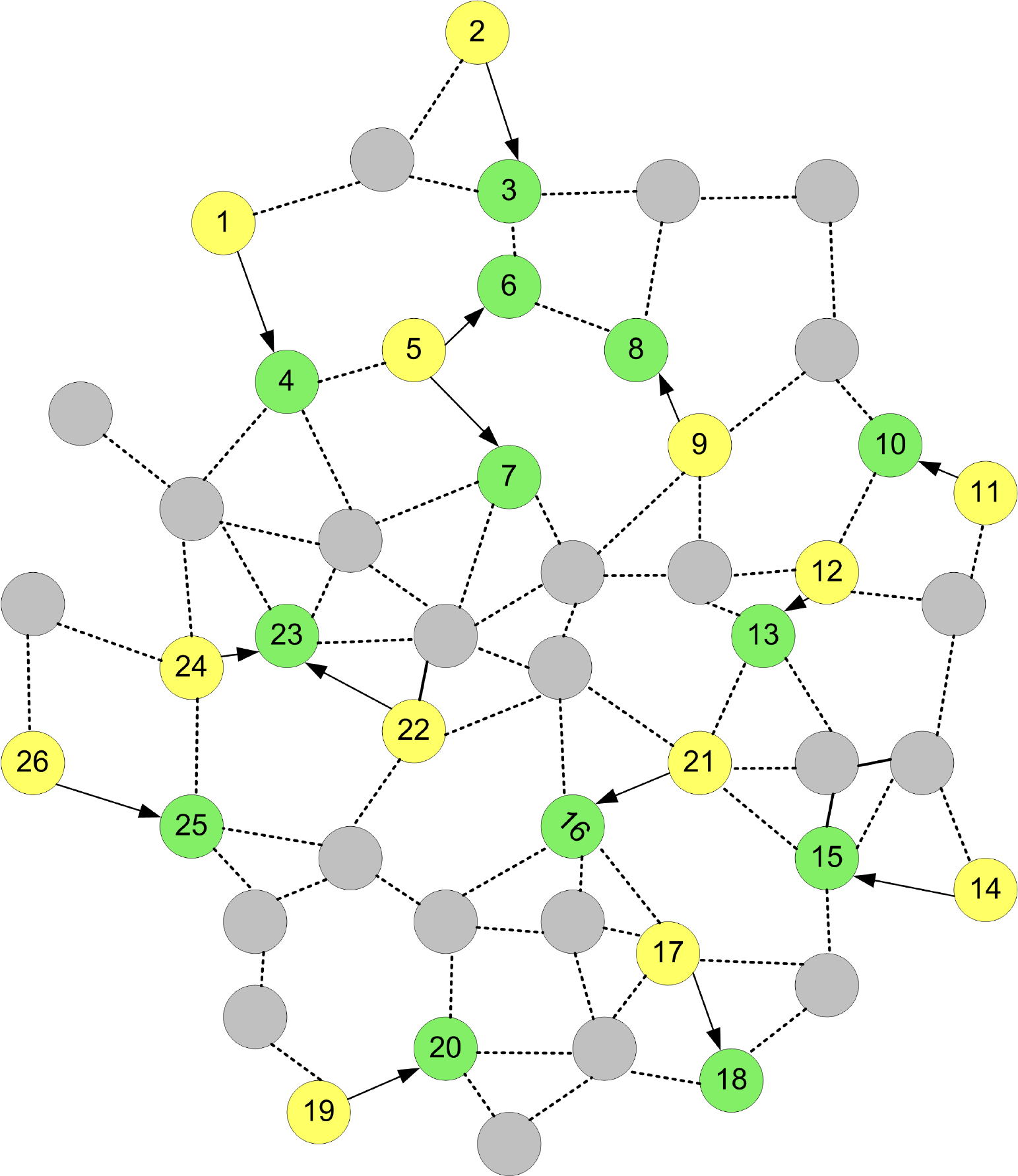}}
\caption{Random geometric graph, yellow and green nodes indicate receivers, transmitters in D2D links.  Gray nodes get their request files from the BS. Arrows show all possible D2D links.}
\label{fig:subfig1}
\end{figure}

\begin{figure}
\centerline{\includegraphics[scale=0.6]{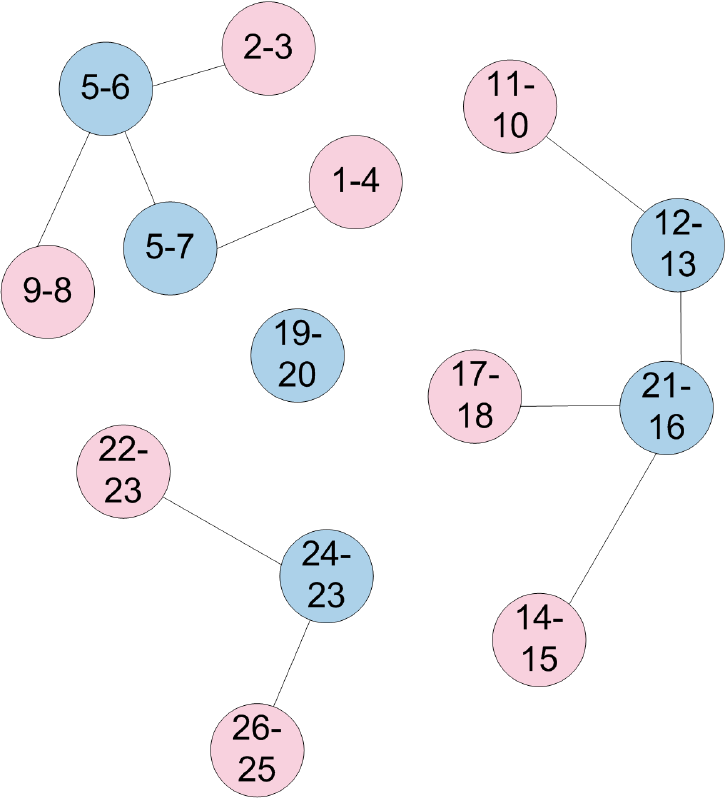}}
\caption{conflict graph based on Figure \ref{fig:subfig1} and one of maximum independent set of the conflict graph; pink vertices are those D2D links that can be activated simultaneously.}
\label{fig:subfig3}
\end{figure}

\section{Analysis}

\label{sec:analysis}

\subsection{Finding the optimal collaboration distance}
\label{Finding the optimal transmission radius}

We are interested in determining the best collaboration distance $r(n)$ and caching policy such that 
the expected number of active D2D links is maximized. 
Our optimization is based on balancing the following tension: The smaller the transmit power, the smaller the region in which a D2D communication creates interference. Therefore, more D2D pairs can be packed into the same area allowing higher frequency reuse.  On the other hand, a small transmit power might not be sufficient to reach a mobile that stores the desired file. Smaller power means smaller distance and hence smaller probability of collaboration opportunities. 

We analyze the case where the nodes do not possess power control with fast adaptation, but rather all users have the same transmit power that depends only on the node density. We then show how to optimize it based on the content request statistics. Our analysis involves finding the best compromise between the number of possible parallel D2D links and the probability of finding the requested content, as discussed above.
Our results consist of two parts. In the first part (upper bound), we find the best achievable scaling for the expected number of active D2D links. In the second part (achievability), we determine an optimal caching policy and $r(n)$ to obtain the best scaling for the expected number of active links $E[L]$. 

The best achievable scaling for the expected number of active D2D links depends on the extend of content reuse. Larger Zipf distribution exponents correspond to more redundancy in the user requests and a small number of files accounts for the majority of video traffic. Thus, the probability of finding requested files through D2D links increases by having access to few popular files via neighbors.

We separate the problem into three different regions depending on the Zipf exponent: $ \gamma_r>1$, $\gamma_r<1$, and $\gamma_r=1$.
For each of these regions, we find the best achievable scaling for $E[L]$ and the optimum asymptotic $r(n)$ denoted by $r_{opt}(n)$. 
We also show that for $\gamma_r>1$ and $\gamma_r<1$ regions a simple distributed caching policy 
has optimal scaling, \textit{i.e.}, matches the scaling behavior that any centralized caching policy could achieve. This caching policy means that each device stores files randomly, with a properly chosen caching distribution, namely a Zipf distribution with parameter $\gamma_c$. For $\gamma_r=1$, we present an optimal centralized caching policy. 

Our first result is the following theorem:
\begin{theorem}\label{theorem_1}
If the Zipf exponent $\gamma_r >1 $,  
\begin{itemize}
\item[i)] \textbf{Upper bound:}
For any caching policy,
 $E[L]=O(n)$,
 \item[ii)] {\textbf{Achievability:}} Given that  $c_1\sqrt{\frac{1}{n}}\leq r_{opt}(n)\leq c_2 \sqrt{\frac{1}{n}}$ \footnote{$c$ and $c_i$s are positive constants that do not depend on $n$.} and using a Zipf caching distribution with exponent $\gamma_c>1$ then $E[L]= \Theta(n)$.
\end{itemize}
\end{theorem}
The first part of the theorem \ref{theorem_1} is trivial since the number of active D2D links can at most scale linearly in the number of users.
The second part indicates that if we choose $r_{opt}(n)=\Theta(\sqrt{\frac{1}{n}})$ and $\gamma_c>1$, $E[L]$ can grow linearly with $n$. There is some simple intuition behind this result: We show that in this regime users are surrounded by a constant number of users in expectation. If the Zipf exponent $\gamma_c$ is greater than one, this suffices to show that the probability that they can find their desired files locally is a non-vanishing constant as $n$ grows. 
Our proof is provided in the Appendix \ref{proof1}.

For the low content reuse region $\gamma_r<1$, we obtain the following result:
\begin{theorem}\label{theorem2}
If $\gamma_r<1$, 
\begin{itemize}
\item[i)] \textbf{Upper bound:} For any caching policy, 
$E[L]=O(\frac{n}{m^{\eta}})$ where $\eta=\frac{1-\gamma_r}{2-\gamma_r}$,
 \item[ii)] \textbf{Achievability:} If $c_3 \sqrt{\frac{m^{\eta+\epsilon}}{n}} \leq r_{opt}(n)\leq c_4 \sqrt{\frac{m^{\eta+\epsilon}}{n}}$ and users cache files randomly and independently according to a Zipf distribution with exponent $\gamma_c$, for any exponent $\eta+\epsilon$,  there exists  $\gamma_c$ such that $E[L]=\Theta(\frac{n}{m^{\eta+\epsilon}})$  where $0<\epsilon<\frac{1}{6}$ and $\gamma_c$ is a solution to the following equation
 \[\frac{(1-\gamma_r)\gamma_c}{1-\gamma_r+\gamma_c}=\eta+\epsilon.\]
\end{itemize}
\end{theorem}
Our proof is provided in the Appendix \ref{proof2}.


We show that when there is low content reuse, linear scaling in frequency re-use is not possible.
At a high level, in order to achieve the optimal scaling, on average a user should be surrounded by $\Theta(m^{\eta})$ users. Comparing with the first region where $\gamma_r> 1$, we can conclude that when there is less redundancy, users have to see more users in the neighborhood to find their desired files locally.

\begin{theorem}\label{theorem_3}
If $\gamma_r=1$
\begin{itemize}
\item[i)] \textbf{Upper bound:}
For any $r(n)$,
 $E[L]=O(\frac{n\log\log(m))}{\log(m)})$ 
  \item[ii)] \textbf{Achievability:} Given that  
 $c_5 \sqrt{\frac{\log(m)}{n\log\log(m)}} \leq r(n)\leq c_6\sqrt{\frac{\log(m)}{n\log\log(m)}}$, there exists a centralized strategy such that 
 \[E[L]= \Theta(\frac{n\log\log(m)}{\log(m)}).\]
\end{itemize}
\end{theorem}

\section{Discussion and Conclusions}

\label{sec:discussion}

As mentioned in Sec. I, the study of scaling laws of the capacity of wireless networks has received significant attention since the pioneering work by Gupta and Kumar~\cite{gupta2000capacity} (\textit{e.g.} see \cite{Tse,grossglauser2001mobility,franceschetti2009capacity}).
 The first result was pessimistic: if $n$ nodes are trying to communicate (say by forming $n/2$ pairs), since the typical distance in a 2D random network will involve  roughly $\Theta(\sqrt{n})$ hops, the throughput per node must vanish, approximately scaling as $1/\sqrt{n}$. There are, of course, sophisticated arguments performing rigorous analysis that sharpens the bounds and numerous interesting model extensions. One that is particularly relevant to this project is the work by Grossglauser and Tse~\cite{grossglauser2001mobility} that showed that if the nodes have infinite storage capacity, full mobility and 
there is no concern about delay, constant (non-vanishing) throughput per node can be sustained as the network scales. 

Despite the significant amount of work on ad hoc networks, there has been very little work on 
file sharing and content distribution over wireless (\cite{femtocaching,chen2011file}) beyond the multiple unicast traffic patters introduced in~\cite{gupta2000capacity}. 
Our result shows that if there is sufficient content reuse, caching fundamentally changes the picture: non-vanishing throughput per node can be achieved, even with constant storage and delay, and without any mobility. 

On a more technical note, the most surprising result is perhaps the fact that in Theorem 2, a simple distributed policy can match the optimal scaling behavior $E[L]=O(\frac{n}{m^{\eta}})$. This means that even if it were possible for a central controller to impose on the devices what to store, the scaling behavior could not improve beyond the random caching policy (though, of course, the actual numerical values for finite device density could be different). Further, for both regimes of $\gamma_r$, the distributed caching policy exponent $\gamma_c$ should not match the request Zipf exponent $\gamma_r$, something that we found quite counter intuitive. 

Overall, even if linear frequency re-use is not possible, we expect the scaling of the library $m$ to be quite small (typically logarithmic) in the number of users $n$. In this case we obtain near-linear (up to logarithmic factors) growth in the number of D2D links for the full spectrum of Zipf exponents.  Our results are encouraging and show that device-based caching and D2D communications can lead to drastic increase of wireless video throughput; and that the benefits increase as the number of participants increases. This in turn implies that the highest throughput gains are achieved in those areas where they are most needed, i.e., where the devices are most concentrated.

\appendices

\section{Proof of Theorem 1}\label{proof1}
The first part of the theorem is easy to see since the number of D2D links cannot exceed the number of users.
Next, we show the second part of the theorem.

For the second part of the theorem, we introduce virtual clusters and we show that the number of virtual clusters that can be potentially active, called \textit{good clusters}, scales like the number of active links. To find the lower bound for good clusters, we limit users to communicate with neighbors in the same cluster. 
Then, we express the probability of good cluster as function of stored files by users within the cluster.
Excluding \textit{self-requests}, \textit{i.e.}, when users find their request files in their own caches, we find a lower bound for good clusters.  We further define  a \textit{value} for each cluster which is the sum of probability of stored files by users. Then we express the probability of goodness as a function of value of clusters.  Using Chernoff bound, we finalize our proof.

\subsection{Active links versus good clusters}
We divide the cell into $\frac{2}{r(n)^2}$ virtual square clusters. 
 Figure \ref{layout} shows the virtual clusters in the cell. The cell side is normalized to $1$ and the side of each cluster is equal to $\frac{r(n)}{\sqrt{2}}$. Thus, all users within a cluster can communicate with each other. Based on our interference model, in each cluster only one link can be activated. When there is an active D2D link within a cluster, we call the cluster \emph{good}. But not all good clusters can be activated simultaneously. According to protocol model, one good cluster can at most block $16$ clusters (see Figure \ref{interfer}). The maximum interference happens when a user in the corner of a cluster transmits a file to a user in the opposite corner. So, we have
 \begin{equation} \label{blocking}
 E[L]\ge \frac{E[G]}{(16+1)}
 \end{equation}
 where $E[G]$ is the expected number of good clusters.  
 
 Since the number of active links scales like the number of good clusters, to prove the theorem it is enough to show that constant fraction of virtual clusters are good.
  This is because $r(n)=\Theta(\sqrt{\frac{1}{n}})$ and there are $\Theta(n)$ virtual clusters in the cell.
  
 \subsection{Limiting users}
   Since we want to find the lower bound for $E[L]$, we can limit users to communicate with users in  virtual clusters they belong to. 
 Hence,
 \begin{align} \label{EG11}
  E[G]
 &\geq \frac{2}{r(n)^2}\sum_{k=0}^{n}{\Pr[\text{good}|k]\Pr[K=k]},
 \end{align}
 where $\frac{2}{r(n)^2}$ is the total number of virtual clusters. $K$ is the number of users in the cluster, which is a binomial random variable with $n$ trials and probability of $\frac{r(n)^2}{2}$, \textit{i.e.,} $K=B(n,\frac{r(n)^2}{2})$.  $\Pr[K=k]$ is the probability that there are $k$ users in the cluster and $\Pr[\text{good}|k]$ is the probability that the cluster is good conditioned on $k$. 
 \subsection{Probability of goodness and stored files}
 To show the result, we should prove that the summation in (\ref{EG11}), \textit{i.e.}, the probability that a cluster is good, does not vanish as $n$ goes to infinity. The probability that a cluster is good depends on what users cache. Therefore,
\begin{align}\label{EG111}
\nonumber  E[G]&\geq \frac{2}{r(n)^2}\sum_{k=0}^{n}{\Pr[K=k]}\\
&\times  \sum_{ \big\{\omega\, \big| |\omega|=k\big\}} \Pr[\text{good}|k,\omega]\Pr[\omega],
 \end{align} 
 where ${{\omega}}$ is a random vector of stored files by users in the cluster and $|\omega|$ denotes the length of vector $\omega$. 
The  $i$th element of $\omega$ denoted by ${\omega_i}\in \{1,2,3,\ldots, m\}$ indicates what user $i$ in the cluster stores.

 \begin{figure}
\centering
\subfigure[]{
\centering \includegraphics[width=4.7cm]{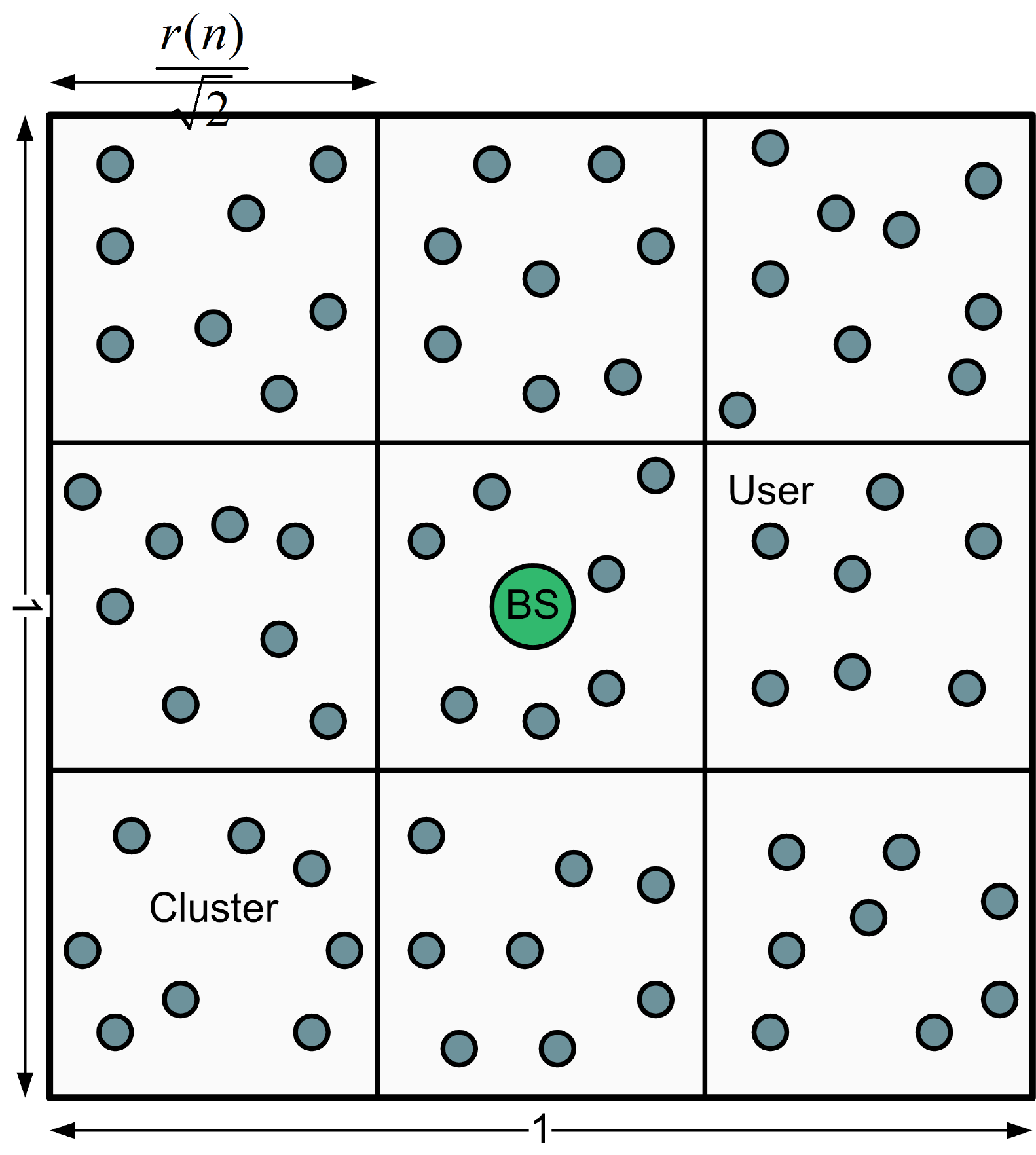}
\label{layout}
}
\subfigure[ ]{
\centering \includegraphics[width=4.7cm]{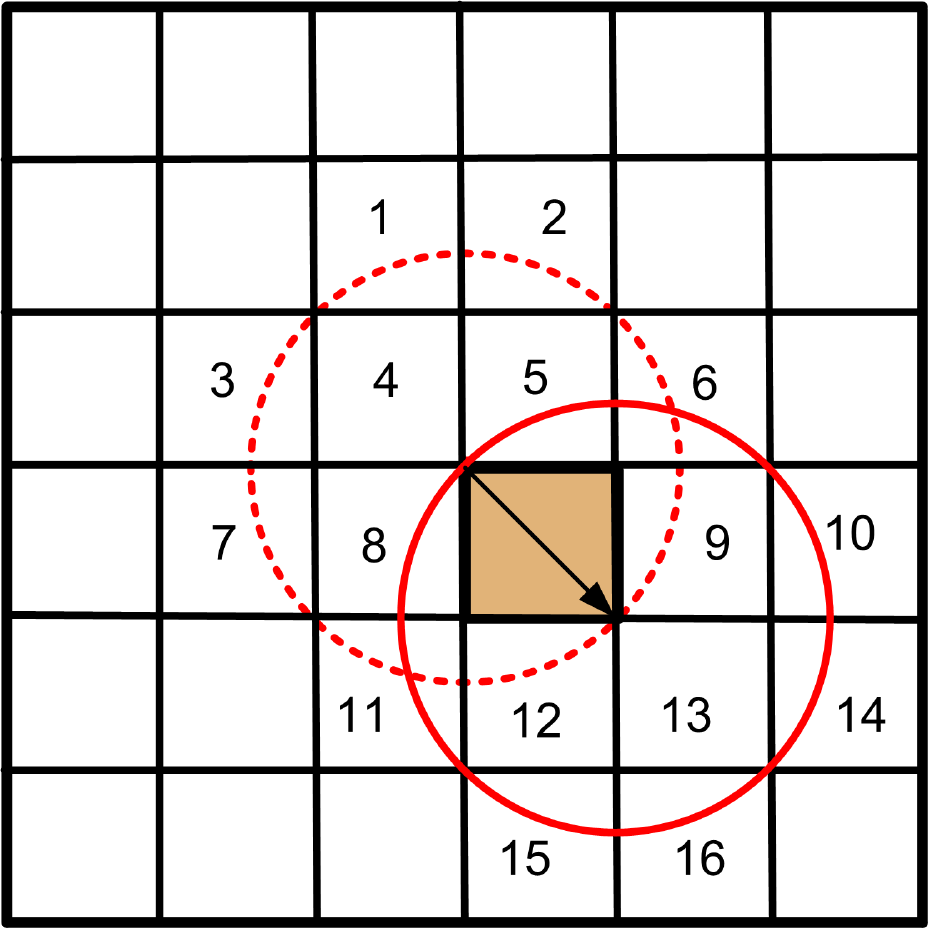}
\label{interfer}
}
\caption{a)~Dividing cell into virtual clusters. b)~In the worst case, a good cluster can block at most $16$ clusters. In the dashed circle, receiving is not possible and in the solid circle, transmission is not allowed.}
\end{figure}

For each $\omega$, we define a value:
\begin{equation}\label{value}
v(\omega)=\sum_{i\in \tilde{ \omega}} f_i,
\end{equation}
where $\tilde{ \omega}=\cup_{j=1}^{|\omega|}\omega_j$ and $\cup$ is the union operation. Actually $v(\omega)$ is the sum of popularities of the union of files in $\omega$.
The cluster is considered to be good if at least a user $i$ in the cluster requests one of the files in $\tilde{ \omega}-\{\omega_i\}$. 
\subsection{Excluding self request}
  A user might find the file it requests in its own cache; in this case clearly no D2D communication will be activated by this user.  We call these cased {\em self-requests}. Accounting for these self-requests,
 the probability that user $i$ finds its request files locally within the cluster is $(v(\omega)-f_{\omega_i})$. Thus, we obtain:
\begin{align}\label{pr[good]111}
 \Pr[\text{good}|k,\omega]
&\ge 1-\big(1-(v(\omega)-\max_{i} f_{\omega_i})\big)^k.
\end{align}
 Let us only consider cases where at least one user in the cluster caches file $1$ (the most popular file). Then, from (\ref{EG111}) and (\ref{pr[good]111}),  the following lower bound is achieved:\\
 \begin{align}
 \nonumber E[G]&\ge \frac{2}{r(n)^2}\sum_{k=1}^{n}{\Pr[K=k]}\\
 &\times \sum_{\omega\in {\bf{x}}}\big[1-(1-(v(\omega)-f_{1}))^k \big]\Pr[\omega].\label{E[G]22}
 \end{align}
 where ${\bf{x}}={ \big\{\omega\, \big| |\omega|=k\, \, \text{and} \,\,1\in \tilde{ \omega}\big \}} $.
\subsection{Probability of goodness and value of clusters}
Instead of taking expectation with respect to $\omega$, we take expectation with respect to $v$, \textit{i.e.,} the value of a cluster.
 Then, 
  \begin{align}
 \nonumber E[G]&\ge  \frac{2}{r(n)^2}\sum_{k=1}^{n}{\Pr[K=k] E_v[1-(1-(v-f_1))^k|A_1^k]}\\
  \nonumber& \ge \frac{2}{r(n)^2}\sum_{k=1}^{n}{\Pr[K=k] E_v[(v-f_1)|A_1^k]},
 \end{align}
 where $A_{1}^k$ is the event that at least one of $k$ users in the cluster caches file $1$ and $E_v[.]$ is the expectation with respect to $v$. Let $A_{1,h}^k$ for $1 \leq h\leq  k$ denote the event that $h$ users out of $k$ users in the cluster cache file $1$. Then, we get:
 \begin{align}\label{E[G]2}
  \nonumber E[G]& \ge \frac{2}{r(n)^2}\sum_{k=1}^{n}\Pr[K=k] \\
  &\times \sum_{h=1}^{k}E_v[(v-f_1)|A_{1,h}^k] \times \Pr[A_{1,h}^k],
 \end{align}
 where $\Pr[A_{1,h}^k]=\left( \begin{array}{c}
k  \\
h \end{array} \right) (p_1)^h
(1-p_1)^{k-h}$ and $p_j$ represents the probability that file $j$ is cached by a user based on Zipf distribution with exponent $\gamma_c$.
 To calculate $E_v[(v-f_1)|A_{1,h}^k]$, we define an indicator function ${\bf 1}_{j}$ for each file $j\ge 2$.  ${\bf 1}_{j}$ is equal to 1 if at least one user in the cluster stores file $j$. Hence,
 \begin {align}
 \nonumber E_v[(v-f_1)|A_{1,h}^k]&=E\big [\sum_{j=2}^m f_j {\bf 1}_{j}|A_{1,h}^k\big]\\
\nonumber & =\sum_{j=2}^m f_j (1-(1-p_j)^{k-h}).
 \end{align}
 \subsection{Chernoff bound}
 To show that the probability of a cluster is good is not vanishing, we use Chernoff bound. First, we limit the interval $k$ to an interval around its average.
 By substituting $E_v[(v-f_1)|A_{1,h}^k]$ in (\ref{E[G]2}), 
 \begin{align}
 \nonumber E[G]
  &\ge \frac{2}{r(n)^2}\sum_{k\in I}\Pr[K=k] \\
 & \times \sum_{h=1}^k  \sum_{j=2}^{m} f_j (1-(1-p_j)^{k-h}) 
 \Pr[A_{1,h}^k],
 \label{E[G]theorem1}
 \end{align} 
 where for any $0<\delta<1$ the interval $I= [nr(n)^2(1-\delta)/2,nr(n)^2(1+\delta)/2]$. Define $k^*\in I$ such that it minimizes the expression in the last line of  (\ref{E[G]theorem1}).
 Since $r(n)=\Theta(\sqrt{\frac{1}{n}})$, $k^*$ is $\Theta(1)$. Then from (\ref{E[G]theorem1}), we have:
 \begin{align}
 \nonumber E[G] & \ge \frac{2}{r(n)^2}\Pr[
 K\in I]  \\
 &\label{before_chernoff1} \times \sum_{h=1}^{k^*}\Big[ \Pr[A_{1,h}^{k^*}]
 \sum_{j=2}^m f_j(1-(1-p_j)^{k^*-h})\Big]\\
&\nonumber  \ge   \frac{2}{r(n)^2} \left(1-2\exp\left({-nr(n)^2\delta^2/6}\right)\right)\\
\times & \sum_{h=k^*p_1(1-\delta_1)}^{k^*p_1(1+\delta_1)}\Big[ \Pr[A_{1,h}^{k^*}] 
\sum_{j=2}^m f_j(1-(1-p_j)^{k^*-h}) \Big],\label{chernoff1}
 \end{align} 
where $0<\delta_1<1$. 
We apply the Chernoff bound in (\ref{before_chernoff1}) to derive (\ref{chernoff1}) \cite{chernoff1952measure}. Since the exponent $nr(n)^2\delta^2/6$ is $\Theta(1)$, we can select the constant $c_1$ such that the term $1-2\exp\left({-nr(n)^2\delta^2/6}\right)$ becomes positive.

Let us define $h^*\in [k^*p_1(1-\delta_1),k^* p_1(1+\delta_1)]$ such that it minimizes the inner summation of  (\ref{chernoff1}), \textit{i.e.,} $\sum_{j=2}^m f_j(1-(1-p_j)^{k^*-h})$.
From (\ref{zipf}), $p_1$ is $\frac{1}{H(\gamma_c,1,m)}$ where function $H$ is defined in lemma \ref{scaling} in Appendix {Some preliminary lemmas}. Lemma \ref{scaling} implies that $p_1=\Theta(1)$ and as a result, $h^*$ is also $\Theta(1)$.
Using the Chernoff bound for random variable $h$ in (\ref{chernoff1}), we get:
\begin{align} \nonumber 
 E[G]&\ge \frac{2}{r(n)^2} \left(1-2\exp\left({-nr(n)^2\delta^2/6}\right)\right) \\
 &\times \left(1-2\exp\left({-k^*p_1\delta_1^2/3}\right)\right)
\sum_{j=2}^m f_j(1-(1-p_j)^{k^*-h^*}).\label{summation}
\end{align}
   $k^*-h^*$  should be greater than $1$ which results in a constant lower bound for $c_1$. The second exponent, \textit{i.e.,} $k^*p_1\delta_1^2/3$ is $\Theta(1)$. Therefore, the term  $\left(1-2\exp\left({-k^*p_1\delta_1^2/3}\right)\right)$ is a positive constant if  $c_1$ is large enough. 
   Further, the summation in (\ref{summation}) satisfies 
\[\sum_{j=2}^m f_j(1-(1-p_j)^{k^*-h^*}) >\sum_{j=2}^m f_jp_j.\]
To show that $E[G]$ scales linearly with $n$, the term  
$\sum_{j=2}^m f_jp_j$ should not be vanishing as $n$ goes to infinity.  
Using part (iv) of lemma \ref{scaling}, we can see that
if $\gamma_r,\gamma_c>1$, 
$\sum_{j=2}^{m}f_j p_j=\Theta(1)$.


\section{Proof of Theorem 2}\label{proof2}
To show the first part of the theorem, like the proof of theorem $1$, we use virtual clusters. We show that the number active links can be at most equal to the number of good clusters. We state the probability of goodness as a function of stored files. To be more precise, we express this probability as a intersection of some decreasing events. Then, we use  FKG inequality, to find an upper bound for probability of goodness. Finally, we divide the whole range of $r(n)$ into four non overlapping regions and show the upper bound for all regions. 
\subsection{Active links versus good clusters}
To show the first part of the theorem, as in proof of the theorem $1$, we divide the cell into $\frac{2}{r(n)^2}$ virtual square clusters. 
  All users within a cluster can communicate with each other. Based on the protocol model, in each cluster only one link can be activated. 
  A stated before, when there is an active D2D link within a cluster, we call the cluster good. In the best case, all the good clusters can be activated simultaneously. Hence,
 \[E[L]\leq E[G],\] 
where $E[G]$ is average number of good clusters. All users can look for their desired files not only in their own clusters but in the caches of all users in their vicinities.  The maximum area that can be covered by all users in a cluster cannot be larger than $\alpha r(n)^2$ where $\alpha \triangleq(\frac{1}{\sqrt{2}}+2)^2$  (the area of dashed square in Figure \ref{cover}). Therefore, 
 \begin{align}\label{EL}
  E[L]\leq \frac{2}{r(n)^2}\sum_{{k}=0}^{n}{\Pr[\text{good}|k]Pr[{K}={k}]},
 \end{align}
 where  
  $K$ is the the number of users in dashed square (called \textit{maximum square}) in Figure \ref{cover} which is 
 binomial random variable with $n$ trials and probability of $\alpha {r(n)^2}$,  ${K}=B(n,\alpha r(n)^2)$. 
  \subsection{Probability of goodness and stored files}
 $\Pr[\text{good}|k]$ is the probability that a cluster is good conditioned on $k$ and it depends on what users in the maximum square stores denoted by $\omega$.
 \[\Pr[\text{good}|k]=\sum_{ \big\{\omega\, \big| |\omega|=k\big\}} \Pr[\text{good}|k,\omega]\Pr[\omega]\]
Let's define an event $A_i(\omega)$ that user $i$  finds its request either in the cache of its neighbors or its own cache.  \begin{align}\label{good}
 \nonumber \Pr[\text{good}|k, \omega]&\leq \Pr[A_1(\omega)\cup A_2(\omega) \cup \ldots \cup A_k(\omega)]\\
 &=1-\Pr[\bar{A}_1(\omega)\cap \bar{A}_2(\omega) \cap \ldots \cap \bar{A}_k(\omega)]
 \end{align}
 Events $A_i(\omega)$ and $A_j(\omega)$ for $j\neq i$ are dependent since they both depend on $\omega$. 
 The probability that event $A_i(\omega)$ happens is:\\
 \[\Pr[A_i(\omega)]=\sum_{j=1}^m f_j {\bf{1}}_j,\]
 where $f_j$ is the probability that user $i$ requests file $j$. ${\bf{1}}_j$ is an indicator function for file $j$ and it is one if file $j\in \omega$. It is easy to check that $v(\omega)$ in (\ref{value}) is equal to $\Pr[A_i(\omega)]$ for any $i$.
 
  \begin{figure}
\centerline{\includegraphics[width=4.7cm]{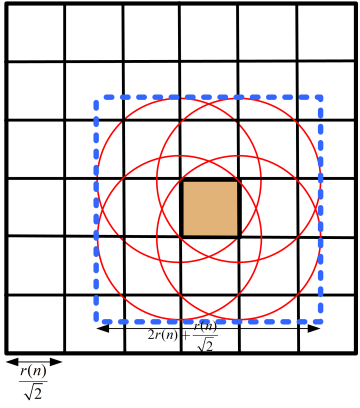}}
\caption{Maximum area covered by all users within a cluster (blue square)}
\label{cover}
\end{figure}

\subsection{Increasing events and FKG inequality}
To find an upper bound for intersection of dependent events
$\bar A_i(\omega)$s in (\ref{good}), we first show that they are decreasing events. Then, we use the FKG inequality for decreasing events \cite{FKG}.
 
  \begin{definition}
 (Increasing event). A random variable X is increasing on $(\Omega,F)$ if $X(\omega) \leq X(\omega^\prime)$ whenever $\omega \leq \omega^\prime$. It is decreasing if $-X$ is increasing. 
  \end{definition}

We assume that $\omega\leq \omega^\prime$ if the value of $\omega$ is less than the value of $\omega^{\prime}$, \textit{i.e.}, \\
\[v(\omega)\leq v(\omega^{\prime}).\]
where the value of $\omega$ is defined in (\ref{value}). Thus, according to this definition, event $A_i(\omega)$ for any $1 \leq i\leq k$ is an increasing event. Applying the FKG inequality for correlated and decreasing events $\bar A_i(\omega)$s \cite{FKG}:
\begin{align}\label{FKG}
\Pr[\bar{A}_1(\omega)\cap \bar{A}_2(\omega) \cap \ldots \cap \bar{A}_k(\omega)]\geq  \Pr[(\bar{A}_1(\omega)]^k.
\end{align}
From (\ref{good}) and (\ref{FKG}), we obtain:\\
 \begin{align}
 \Pr[\text{good}|k,\omega]&\leq 1-\Pr[(\bar{A}_1(\omega)]^k \\
& \leq 1-(1-\sum_{j=1}^k f_j)^k. \label{max A}
 \end{align} 
To derive (\ref{max A}), we used the fact that the probability of event $A_1(\omega)$ is maximized if the $k$ most popular files  is in $\omega$. The obtained upper bound in (\ref{max A}) does not depend on $\omega$.  Hence,
\begin{align}\label{good1}
\Pr[\text{good}|k]\leq 1-(1-\sum_{j=1}^k f_j)^k.
\end{align}
In the following, we will consider four non overlapping regions for $r(n)$ and for each region, we will prove the first part of the theorem.

\subsection{First region}
We first consider the region $r(n)=O(\sqrt{\frac{1}{n}})$.
From (\ref{EL}) and (\ref{good1}),
\begin{subequations}
\begin{align}
  \label{EL_up} E[L]&\leq \frac{2}{r(n)^2}\sum_{{k}=0}^{n}{\big[1-(1-\sum_{j=1}^k f_j)^k\big]\Pr[{K}={k}]}\\
  &\leq \frac{2}{r(n)^2}\sum_{{k}=1}^{n}{k\Pr[{K}={k}]} \sum_{j=1}^k f_j.  \label{EL_up_2} 
 \end{align}
 \end{subequations}
 Using part (iii) of lemma \ref{scaling}, the second summation $\sum_{j=1}^k f_j \leq 2 \frac{k^{1-\gamma_r}}{m^{1-\gamma_r}}$. Thus, 
 \begin{subequations}
\begin{align}
\nonumber E[L] &\le \frac{4}{{{r(n)^2}}}\sum\limits_{k = 0}^n {\frac{{{k^{2 - \gamma_r}}}}{{{m^{1 - \gamma_r }}}}\Pr [K = k]} \\
\nonumber  &\le \frac{4}{{{r(n)^2}{m^{1 - \gamma_r }}}}\sum\limits_{k = 0}^n {{k^2}\Pr [K = k]} \\
\nonumber &= \frac{4}{{{r(n)^2}{m^{1 - \gamma_r }}}}E[{K^2}].
\end{align}
\end{subequations}
For the Binomial random variable $K=B(n,\alpha r(n)^2)$,
\[E[{K^2}]=(\alpha nr(n)^2)^2+\alpha nr(n)^2(1-\alpha r(n)^2)\]
 Therefore,
\begin{align}
\nonumber E[L] &\le \frac{4}{{{r(n)^2}{m^{1 - \gamma_r}}}}\left( (\alpha nr(n)^2)^2+\alpha nr(n)^2(1-\alpha r(n)^2)\right)\\
\nonumber  &= \frac{{4n}}{{{m^{1 - \gamma_r }}}}\left( \alpha^2{n{r(n)^2} + \alpha (1 - \alpha{r(n)^2})} \right)\\
\nonumber  &= c \frac{{n}}{{{m^{1 - \gamma_r }}}}.
 \end{align}
 where $c$ is some constant.
\subsection{Second region}
Then, we consider the region that $r(n)=\Omega(\sqrt{\frac{1}{n}})$, and $r(n)= O(\sqrt{\frac{\log(m)}{n}})$.
Equation (\ref{EL_up}) implies:
\begin{align}\label{first lemma5}               
 \nonumber E[L] &\leq  \frac{2}{{{r(n)^2}}}\sum\limits_{0 \le k <k_0} {[1-(1-\sum_{j=1}^k f_j)^k]\Pr [K = k]} \\
  &+ \frac{2}{{{r(n)^2}}}\sum\limits_{k \geq k_0} {[1-(1-\sum_{j=1}^k f_j)^k]\Pr [K = k]}.
\end{align}
Assuming that $r(n)\leq \sqrt{\frac{c\log(m)}{n}}$, we choose $k_0= 6\alpha c \log (m)$ where $c$ is some constant. Note $[1-(1-\sum_{j=1}^k f_j)^k]$ is an increasing function of $k$ and it is less and equal to $1$. 
Therefore, (\ref{first lemma5}) implies,
\begin{align}
\nonumber E[L]  &\le \frac{2}{{{r(n)^2}}} [1-(1-\sum_{j=1}^{k_0} f_j)^{k_0}] \Pr[K<k_0]\\
 &+ \frac{2}{{{r(n)^2}}}\Pr[K\geq k_0],\\
 &\le \frac{2}{{{r(n)^2}}} [{k_0}\sum_{j=1}^{k_0} f_j] \Pr[K<k_0]
 + \frac{2}{{{r(n)^2}}}\Pr[K\geq k_0]\\
 &\le \frac{4}{{{r(n)^2}}} \big[\frac{{k_0}^{2-\gamma_r}}{m^{1-\gamma_r}}\big] \Pr[K<k_0]
 + \frac{2}{{{r(n)^2}}}\Pr[K\geq k_0]\label{writting as 2summation}
\end{align}
We use part (iii) of lemma \ref{scaling} to derive the last equation. For the binomial random variable $K$ and for any $R \geq 6 E[K]$, the  Chernoff bound holds\cite{chernoff1952measure}:\\
\begin{equation}\label{Chernoff_bound1}
\Pr [K \geq R] \le 2^{-R}.
\end{equation}
Applying the Chernoff bound and substituting $k_0$ in (\ref{writting as 2summation}), we acquire:
\begin{align}\label{new_25}
\nonumber E[L] &\le \frac{4}{{{r(n)^2}}}\frac{{(6\alpha c \log (m))^{2 - \gamma }}}{{{m^{1 - \gamma }}}}\\
\nonumber  &+ \frac{2}{{{r(n)^2}}}{2^{- 6\alpha c \log (m)}}\\
 &= {4({6\alpha c})^{2-\gamma}} \frac{1}{{r(n)^2}}\frac{{{{\left( {\log (m)} \right)}^{2 - \gamma }}}}{{{m^{1 - \gamma }}}} + \frac{1}{{{r(n)^2}}}\frac{2}{{{m^{6\alpha c\log 2} }}}.
\end{align}
The function $f(x)=\frac{\log(x)}{x^\beta}$ is always less than $\frac{1}{\beta}$ where $\beta>0$. Thus, $\log(m)\leq \frac{m^{{\eta}^2}}{{\eta}^2}$. 
\begin{align}
\nonumber E[L] &\le {4({6\alpha c})^{2-\gamma}} \frac{1}{{{r(n)^2}}}{\left( {\frac{{{m^{{\eta ^2}}}}}{{{\eta ^2}}}} \right)^{2 - \gamma }}\frac{1}{{{m^{1 - \gamma }}}} + \frac{2n}{{{m^{\eta} }}}\\
\nonumber  &= \frac{4({6\alpha c})^{2-\gamma}}{{\eta ^{4 -2 \gamma }}} \frac{1}{r(n)^2{{m^\eta }}} + \frac{2n}{{{m^\eta }}}\\
 &=\Theta( \frac{n}{m^{\eta}})
\end{align}
\subsection{Third and fourth regions}
For the third region, $r(n)= \Omega(\sqrt{\frac{\log(m)}{n}})$ and $r(n)=O(\sqrt{\frac{1}{n}})$. To show the upper bound for $E[L]$ in this region, we follow similar procedure in the second region by setting $k_0=6\alpha nr(n)^2$.
For the last region $r(n)=\Omega(\sqrt{\frac{m^{\eta}}{n}})$, the total number of all virtual clusters $\frac{2}{r(n)^2}=O(\frac{n}{m^{\eta}})$. Thus, for this range of $r(n)$, $E[L]=O(\frac{n}{m^{\eta}})$.
\\
\\

In the following, we will show the \textit{second part} of the theorem. Similar to proof of the theorem $1$, we relate the number of good clusters and active links. We restrict users to communicate with their neighbors in their clusters. We further limit users not to get certain files from their neighbors although some neighbors might store these files. In this case the value of a cluster is the sum of probability of stored files that users can get via their neighbors.  By the restriction on files  the value of cluster becomes concentrated around its mean. We also consider self requests in finding the lower bound.  Applying Chernoff bound and Azuma inequality we show that the probability of goodness is not vanishing when a user is surrounded in average by $\pi nr_{opt}(n)^2$ neighbors from which the result follows.

Define $\eta_1\triangleq \eta+\epsilon=\frac{(1-\gamma_r)\gamma_c}{1-\gamma_r+\gamma_c}$. We should show that if we choose $r(n)=\Theta(\sqrt{\frac{m^{\eta_1}}{n}})$, the probability that a virtual cluster is good does not vanish as $n$ grows.

When $r(n)=\Theta(\sqrt{\frac{m^{\eta_1}}{n}})$, there are $\Theta(\frac{n}{m^{\eta_1}})$ virtual clusters.  The number of active D2D links is upper bounded by the number of virtual clusters. Thus, $E[L]=O(\frac{n}{m^{\eta_1}})$.
Then, we show that for $c_3 \sqrt{\frac{m^{\eta_1}}{n}}\leq r(n)\leq c_4\sqrt{\frac{m^{\eta_1}}{n}}$, $E[L]=\Omega(\frac{n}{m^{\eta_1}})$. 
To do this, we follow similar procedure in theorem \ref{theorem_1}. We divide the cell into virtual clusters and we allow each user to look for its desired file just within its cluster. As mentioned before, each cluster can block at most $16$ other clusters (Figure \ref{interfer}).

\subsection{Limiting users and excluding self request}
To find the lower bound, we even more restrict users. We assume that users can not get  files $\{1,2,\ldots, q-1\}$  locally even if there are users in the cluster that cache these files where $q=m^{\frac{\eta_1}{\gamma_c}}$. So, caching files  $\{1,2,\ldots, q-1\}$ doesn't have any value for any user in the cluster. 

 $E[L]$ is lower bounded by expression in (\ref{blocking}) where the lower bound for $E[G]$ is given in (\ref{EG111}). Similar to (\ref{pr[good]111}), we exclude the self requests . Thus, the probability that a cluster is good conditioned on $k$ and $\omega$ is
 \begin{align}\label{pr[good]}
\Pr[\text{good}|k,\omega]
&\ge 1-\left(1-\left(v(\omega)-\max_{i\in\{q,\ldots,m\}} f_{\omega_i}\right)\right)^k
\end{align}
 where $v(\omega)=\sum_{j=q}^m f_j{\bf 1}_{j}$ and ${\bf 1}_{j}$ is an indicator function. ${\bf 1}_{j}$ is one if at least one user in the virtual cluster stores file $j$. We limit ourselves to all cases in which at least one user caches file $q$. Hence,
\begin{align}\label{good3}
 \Pr[\text{good}|k,\omega]&\ge 1-\left(1-\left(v(\omega)-f_{q}\right)\right)^k
\end{align}
\subsection{Chernoff bound}
As in proof of theorem $1$, we first limit the interval of $k$ and the we use the Chernoff bound.
 
 By restricting $k$ to an interval around its average, \textit{i.e.}, $I=[nr(n)^2(1-\delta)/2, nr(n)^2(1+\delta)/2]$ where $0<\delta<1$ and applying (\ref{good3}) in (\ref{EG111}), the following lower-bound is obtained:
 \begin{align}
 \nonumber E[G]&\ge \frac{2}{r(n)^2}\sum_{k \in I}Pr[K=k]\\
 & \times \sum_{ \omega\in {\bf{x}}} \big[1-(1-(v(\omega)-f_{q}))^k]\Pr[\omega],
 \end{align}
 where ${\bf{x}}={ \big\{\omega\, \big| |\omega|=k\, \, \text{and} \,\,q\in \omega\big \}}$.
  Let $k^{*}$ be 
 \begin{align}
 \nonumber k^{*}&\triangleq \arg\min_{k \in I }\sum_{\omega \in{\bf{x}} } \big[1-(1-(v(\omega)-f_{q}))^k]\Pr[\omega]
 \end{align}
Notice that $k^{*}$ and also all $k\in I $ are $\Theta(m^{\eta_1})$. Then,
\begin{align}
\nonumber E[G]&\ge \frac{2}{r(n)^2}\Pr[K \in I]\\
\nonumber& \times   \sum_{ \omega  \in {\bf{x}}} 1-(1-(v(\omega)-f_{q}))^{k^{*}}\Pr[\omega]\\
 &\nonumber \geq  \frac{2}{r(n)^2} \left(1-2\exp\left({-nr(n)^2\delta^2/6}\right)\right)\\
 &\times \sum_{ \omega \in {\bf{x}}} \big[1-(1-(v(\omega)-f_{q}))^{k^{*}}\big]\Pr[\omega].\label{use_chernoff}
\end{align}
We use the Chernoff bound in (\ref{use_chernoff}).
Let $A_{q,h}^{k}$ denote the event that $1\leq h \le k$ users cache file $q$. Then, we can rewrite the above lower-bound as follows:
\begin{align}
 \nonumber E&[G]\geq  \frac{2}{r(n)^2} \left(1-2\exp\left({-nr(n)^2\delta^2/6}\right)\right)\\
 \nonumber &\times \sum_{h=1}^{k^{*}}E_v[1-(1-(v-f_{q}))^{k^{*}}|A_{q,h}^{k^*}] \Pr[A_{q,h}^{k^*}]\\
 \nonumber & \geq  \frac{2}{r(n)^2} \left(1-2\exp\left({-nr(n)^2\delta^2/6}\right)\right)\\
  &\times \sum_{h=k^{*}p_q(1-\delta_1)}^{k^{*}p_q(1+\delta_1)}E_v[1-(1-(v-f_{q}))^{k^{*}}|A_{q,h}^{k^*}] \Pr[A_{q,h}^{k^*}]
\end{align}
where $\Pr[A_{q,h}^{k^*}]= \left( \begin{array}{c}
k^{*}  \\
h \end{array} \right) ({p_q})^{h}\left(1-{p_q}\right)^{k^{*}-h}$, $k^*p_q$ is the average of binomial random variable $h$ and  $0<\delta_1<1$. 

Define $h^*$ as
\begin{align}
\nonumber h^*&\triangleq \arg\min_{k^{*}p_q(1-\delta_1) \leq h\leq k^{*}p_q(1+\delta_1)} \\&E_v[1-(1-(v-f_{q}))^{k^{*}}|A_{q,h}^k] 
\end{align}
 Using Chernoff bound for binomial random variable $h$, we obtain:
\begin{align}\label{EG in achievebility}
E[&G] \geq  \frac{2}{r(n)^2} \left(1-2\exp\left({-nr(n)^2\delta^2/6}\right)\right) \\
&\times\left(1-2\exp\left({k^{*}p_q\delta_1^2/3}\right)\right)
 \nonumber E_v[1-(1-(v-f_{q}))^{k^{*}}|A_{q,h^*}^{k^*}]
\end{align}
The probability that a user caches file $q$ is:
\begin{align}
\nonumber p_q&=\frac{\frac{1}{q^{\gamma_c}}}{\sum_{j=1}^{m}\frac{1}{j^{\gamma_c}}}\\
&=\frac{\Theta(\frac{1}{m^{\eta_1}})}{H(\gamma_c,1,m)} \label{p_q}
\end{align}
where function $H$ is defined in lemma \ref{scaling}. We show in lemma \ref{scaling} that $H(\gamma_c,1,m)=\Theta(1)$ given that $\gamma_c >1$. Thus, all $h\in[k^{*}p_q(1-\delta_1),k^{*}p_q(1+\delta_1)]$ are $\Theta(1)$. By selecting the constant  $c_3$ large enough, the second exponential term $\left(1-2\exp\left({-k^{*}p_q\delta_1^2/3}\right)\right)$ will be greater than zero.
\subsection{Probability of goodness is not vanishing }
To complete the proof, it is enough to show that the probability that cluster is good, \textit{i.e.}, $E_v[1-(1-(v-f_{q}))^{k^{*}}|A_{q,h^*}^{k^*}]$ given in (\ref{EG in achievebility})
  does not vanish.
\begin{align}\label{fv}
\nonumber &E_v[1-(1-(v-f_{q}))^{k^{*}}|A_{q,h^*}^{k^*}] \geq \\
 &
\int_{|v-E_v[v|A_{q,h^{*}}^{k^*}]|<t}(1-(1-(v-f_{q}))^{k^{*}})f_{v|A_{q,h^{*}}^{k^*}}(v) dv
\end{align}
where $f_{v|A_{q,h^{*}}^{k^*}}(v)$ is a probability distribution function of value $v$ conditioned on $A_{q,h^{*}}^{k^*}$ and $0<t<E_v[v|A_{q,h^{*}}^{k^*}]$. The average of $v$ conditioned on $A_{q,h^{*}}^{k^*}$ is given by:
\begin{align}\label{E_v[v|A_{q,h}^{k}]}
E_v[v|A_{q,h^{*}}^{k^*}]=f_q+\sum_{j=q+1}^{m} f_j (1-(1-p_j)^{k^{*}-h^{*}})
\end{align}
From equation  (\ref{fv}) and since $(1-(1-(v-f_{q}))^{k^{*}})$ is an increasing function of $v$,
\begin{subequations}
\begin{align}
\nonumber E_v[1-&(1-(v-f_{q}))^{k^{*}}|A_{q,h^*}^{k^*}]  \\
\nonumber &\geq \left(1-\left(1-(\left(E_v[v|A_{q,h^{*}}]-t\right)-f_{q})\right)^{k^{*}}\right)\\
\times & \Pr \big [{|v-E_v[v|A_{q,h^*}^{k^*}]|<t}\big]\\
\nonumber &\geq 1-\exp(-{k^{*}}(E_v[v|A_{q,h^*}^{k^*}]-t-f_{q}))\\
 \times & \Pr \big[{|v-E[v|A_{q,h^*}^{k^*}]|<t}\big] \label{concentration}
\end{align}
\end{subequations}
We show in lemma \ref{E[v]scaling}, $E_v[v|A_{q,h^{*}}^{k^*}]=\Theta(\frac{1}{m^{\eta_1}})$. Thus, $k^*E_v[v|A_{q,h^*}^{k^*}]=\Theta(1)$. Furthermore, (\ref {fv}) implies 
\[t=O(E_v[v|A_{q,h^*}^{k^*}])=O(\frac{1}{m^{\eta_1}}).\]
 Similar to (\ref{p_q}), we can show that $f_q=\Theta(\frac{1}{m^{\eta_2}})=O(E_v[v|A_{q,h^*}^{k^*}])$ where  $\eta_2=\frac{(1-\gamma_r)(1+\gamma_c)}{1-\gamma_r+\gamma_c}$. Thus, the exponent in the first term of (\ref{concentration}) is $\Theta(1)$. 
To prove the result, it is enough to show the second term in (\ref{concentration}) does not approach zero as $n$ grows. By applying the Azuma Hoeffding inequality in lemma {\ref{azuma}},
\begin{align}\label{azuma1}
\Pr\big[{|v-E_v[v|A_{q,h^*}^{k^*}]|\leq t}\big]\geq  1-2\exp\big(-\frac{2t^2}{(k^*-h^*)(f_q)^2}\big) 
\end{align}
Due to the fact that $k^*=\Theta(m^{\eta_1})$ and $h^{*}=\Theta(1)$, the term $k^*-h^*=\Theta(m^{\eta_1})$. 
If we select $t=\Theta(\frac{1}{m^{\eta_1}})$, we can observe that the exponent $\frac{2t^2}{(k^*-h^*)(f_q)^2}$ scales with $m^{\eta_2(2-\gamma_c)}$. Hence, if $\gamma_c<2$, the exponent goes to infinity as $n$ grows. $\gamma_c<2$ implies that $\epsilon<\frac{1}{6}$.
This means that $v$ is concentrated around its average with high probability if  $\epsilon \le \frac{1}{6}$ and as a result, the second term in (\ref{concentration}) is positive constant when $n$ goes to infinity.

\section{Proof of Theorem 3}\label{proof3}

The proof of the first part of the theorem is similar to the proof of the theorem 2. $E[L]$ is upper bounded by the expression in (\ref{EL_up}). Next, we consider three non-overlapping regions for $r(n)$ and we show the upper bound is valid for every $r(n)$.
\subsection{First region}
First, we assume $r(n)=O(\sqrt{\frac{\log\log(m)}{n}})$.  Eqation (\ref{EL_up_2}) and part (v) of lemma \ref{scaling} imply, 
\begin{subequations} 
\begin{align}
E[L]& \leq \frac{2}{r(n)^2}\sum_{{k}=0}^{n}\frac{\log(k)+1}{\log(m)}\Pr[{K}={k}]\\
&\leq \frac{2}{r(n)^2\log(m)}\sum_{{k}=0}^{n}k^2\Pr[{K}={k}]\\
 & = \frac{2}{r(n)^2\log(m)} E[K^2]\\
\label{Ek2}&\leq \frac{2}{r(n)^2\log(m)} \big[(\alpha nr(n)^2)^2+\alpha nr(n)^2\big]\\
&=\frac{2n}{\log (m)} \big[{\alpha^2 n{r(n)^2} + \alpha}\big]\\
\label{upp1}&\leq \frac{2cn}{{\log (m)}}(\alpha^2 \log\log(m) + \alpha)\\
&=\Theta(\frac{n\log\log(m)}{\log(m)})
\end{align}
\end{subequations}
To derive (\ref{upp1}), we use the range of $r(n)$.

\subsection{Second  and third regions}
Let's consider the second region for $r(n)$. In this region $r(n)=\Omega(\sqrt{\frac{\log\log(m)}{n}})$ and $r(n)=O(\sqrt{\frac{\log(m)}{n}})$. From (\ref{EL_up}),
\begin{align}
\nonumber E[L]&\leq \frac{2}{r(n)^2}\sum_{{k}=0}^{6\alpha nr(n)^2}[{1-(1-\sum_{j=1}^k f_j)^k]\Pr[{K}={k}]}\\
&+ \frac{2}{r(n)^2}\sum_{{k}=6\alpha nr(n)^2}^{n}[{1-(1-\sum_{j=1}^k f_j)^k]\Pr[{K}={k}]}
\end{align}
where $\alpha$ is defined in theorem $2$. The term $[1-(1-\sum_{j=1}^k f_j)^k]$ is an increasing function of $k$, thus,
\begin{subequations}
\begin{align}
\nonumber E[L]&\leq \frac{2}{r(n)^2}\left[{1-\left(1-\sum_{j=1}^{6\alpha nr(n)^2} f_j\right)^{6\alpha nr(n)^2}}\right]\\
&+ \frac{2}{r(n)^2}\Pr[K> 6\alpha nr(n)^2]\\
& \leq  \frac{2}{r(n)^2}  6\alpha nr(n)^2 \sum_{j=1}^{6\alpha nr(n)^2} f_j+  \frac{2}{r(n)^2} {2^{ - 6\alpha n{r(n)^2}}} \label{chernof}
\end{align}
\end{subequations}
In (\ref{chernof}), we applied the Chernoff bound \cite{chernoff1952measure}. From lemma \ref{scaling} and the range of $r(n)$, we obtain
\begin {align} 
\nonumber E[L]& \le 12 \alpha n\frac{{\log (6\alpha n{r(n)^2}) + 1}}{{\log (m)}} + \frac{2}{r(n)^2} {2^{ - 6\alpha n{r(n)^2}}}\\
\nonumber &\leq 12 \alpha n\frac{{\log (6\alpha c_7 \log(m)) + 1}}{{\log (m)}}\\
\nonumber &+ \frac{2n}{c_8 \log\log(m)} {2^{ - 6\alpha c_8 \log\log(m)}}\\
\nonumber &=\Theta(\frac{n\log\log(m)}{\log(m)})+ \frac{2n}{{c_8\log\log (m)}}\times \frac{1}{ {\log(m)}^{6\alpha c_8 \log(2)}}\\
\nonumber &=\Theta(\frac{n\log\log(m)}{\log(m)}).
\end{align}

For the last region, \textit{i.e.}, $r(n)=\Omega(\sqrt{\frac{\log(m)}{n}})$, the total number of virtual clusters is $O(\frac{n}{\log(m)})$  and as a result, $E[L]=O(\frac{n}{\log(m)})=O(\frac{n\log\log(m)}{\log(m)})$.
\\
\\

In the following, we will show the \textit{second part} of the theorem. We propose a centralized algorithm that can match the upper bound. The BS divides the cell into virtual cluster of size $r(n)=\Theta(\sqrt{\frac{\log(m)}{n\log\log(m)}})$. Given that there are $k$ users in a cluster, each of them should cache one of the $k$ most popular files.  We show that under this caching policy, we can match the upper bound.
To find the lower bound, we assume that users can just find their desired files just within clusters they belong to. The lower bound for $E[L]$ and  $E[G]$ are respectively given in (\ref{blocking}) and (\ref{EG11}). Limiting the range of $k$ results in
\begin{align} \label{EG22} 
  E[G]
 &\geq \frac{2}{r(n)^2}\sum_{k\in I}{\Pr[\text{good}|k]\Pr[K=k]},
 \end{align}
 where $I=[nr(n)^2(1-\delta)/2,nr(n)^2(1-\delta)/2]$. Under this centralized caching policy the value of stored files within a cluster with $k$ users is $v(k)=\sum_{j=1}^k f_j$. The cluster is good if at least one user within a cluster requests one of the $k$ most popular files not stored in its own cache\\
 \begin{align}
 \nonumber \Pr[\text{good}|k] &\geq 1 - \left(1-\left(v(\omega)-f_1\right)\right)^k\\
 \nonumber &=1 - \left(1-\sum_{j=2}^k f_j\right)^k\\
\nonumber  & \ge 1-\exp\left(-k \sum_{j=2}^k f_j\right)\\
 & \geq 1 - {\exp{ \left(-\frac{{k\left(\log (k)-1\right)}}{{\log (m) + 1}}\right)}} \label{exp}
 \end{align}
 In the last equation we used lemma \ref{scaling}. The expression in (\ref{exp}) is an increasing function of $k$. Thus, (\ref{exp}) and (\ref{EG22}) imply
 \begin{align}
 E[G]&\geq \left(1 - {\exp{ \left(-\frac{{k_{min}\left(\log (k_{min})-1\right)}}{{\log (m) + 1}}\right)}}\right) \Pr[K \in I]\\
\nonumber  &\geq\Big(1 - {\exp{ \big(-c\frac{{k_{min}\log (k_{min})}}{{\log (m)}}\big)}}\Big)\\
 &\times \left(1-2\exp\left({-nr(n)^2\frac{\delta^2}{6}}\right)\right) \label{EG33}
 \end{align}
 where $k_{min}=nr(n)^2(1-\delta)/2=\Theta(\frac{\log(m)}{\log\log(m)})$. We use the Chernoff bound to derive (\ref{EG33}).
  As $n$ grows, the second term in (\ref{EG33}) goes to $1$. It can be seen that the first term in (\ref{EG33}) is also $\Theta(1)$. Thus, $E[G]$ and consequently $E[L]$ are $\Theta(\frac{n\log\log(m)}{\log m})$.

\section{Some preliminary lemmas}\label{Some preliminary lemmas}
\begin{lemma} \label{scaling}
\begin{itemize}
\item[i)] If $\gamma>1$ and $a=o(b)$,
$H( \gamma,a,b)=\Theta(\frac{1}{a^{\gamma-1}})$.
\item[ii)] If $ \gamma<1$, $a=o(b)$, and $a=\Theta(1)$,
$H(\gamma,a,b)=\Theta({b^{1-\gamma}})$.
\item[iii)] if $\gamma_r<1$, $\sum_{j=1}^k f_j \leq 2 \frac{k^{1-\gamma_r}}{m^{1-\gamma_r}}$.
\item[iv)] If $\gamma_c, \gamma_r>1$, $\sum_{i=2}^{m}f_i p_i=\Theta(1)$.
\item[v)] if $\gamma=1$, $\sum_{j=l}^{k}f_j\leq \frac{{\log (k) + 1}}{{\log (m)}}$ and $\sum_{j=2}^{k}f_j \geq \frac{{\log (k) - 1}}{{\log (m) + 1}} $.
\end{itemize}
where $H(\gamma,a,b)=\sum\limits_{j=a}^{b}\frac{1}{{{i^\gamma }}}$,
\begin{equation}
 p_i=\frac{{\frac{1}{{{i^{\gamma_c} }}}}}{{\sum\limits_{j = 1}^m {\frac{1}{{{j^{\gamma_c} }}}} }},\,\,\ 1\leq i\leq m.
 \end{equation}
 and $f_i$ is defined in (\ref{zipf}).

\end{lemma}

\begin{proof}
We first prove the parts (i) and (ii) of the lemma.

  $\frac{1}{x^{\gamma}}$ is monotonically decreasing. Thus,\\
    \begin{equation}\label{first_ineq}
    H(\gamma,a,b)\ge \int\limits_{x=a}^{b} {\frac{1}{x^{\gamma}}} =\frac{{b^{(-\gamma+1)}-a^{-\gamma+1}}}{-\gamma+1}
    \end{equation}
      We also have the following inequality:
      \begin{align}
    \nonumber H(\gamma,a,b)-\frac{1}{a^{\gamma}}&=\sum\limits_{j=a+1}^{b}{\frac{1}{j^{\gamma}}}\\
    &\leq \int\limits_{x=a}^{b}{\frac{1}{x^{\gamma}}}=\frac{{b^{(-\gamma+1)}-a^{(-\gamma+1)}}}{-\gamma+1} \label{second_ineq}
      \end{align}
  Thus, $H(\gamma,a,b)$ satisfies:
   \begin{equation}\label{bounds}
        \frac{{b^{(-\gamma+1)}-a^{-\gamma+1}}}{-\gamma+1} \leq H(\gamma,a,b) \leq \frac{{b^{(-\gamma+1)}-a^{-\gamma+1}}}{-\gamma+1}+\frac{1}{a^{\gamma}}
      \end{equation} 
      Therefore, if $\gamma>1$, $H(\gamma,a,b)=\Theta(\frac{1}{a^{\gamma-1}})$. Besides, if $\gamma<1$ and $a=\Theta(1)$, then $H(\gamma,a,b)=\Theta({b^{1-\gamma}})$. 
     
      For part (iii), using (\ref{first_ineq}) and  (\ref{second_ineq}), we have
      \begin{align}
 \nonumber \sum_{j=1}^k f_j&=\frac{H(\gamma_r,1,k)}{H\gamma_r,1,m} \\
 \le & \frac{{{k^{(1 - \gamma_r )}} - \gamma_r }}{{m{{}^{(1 - \gamma_r )}} - 1}} \nonumber\\
  \nonumber \le & 2\frac{{{k^{(1 - \gamma_r )}}}}{{m{{}^{(1 - \gamma_r )}}}}.
     \end{align}
      
      Next we show part (iv). From (\ref{zipf}), we have:
 \begin{align}
 \nonumber \sum_{j=2}^{m}f_j p_j = & \frac{\sum_{j=2}^m\frac{1}{j^{\gamma_r+\gamma_c}}}{\sum_{j=1}^m \frac{1}{j^{\gamma_r}}\sum_{j=1}^{m} \frac{1}{\gamma_c}}\\
  =&\frac{ H( \gamma_c+\gamma_r,2,m)}{H(\gamma_c,1,m)H(\gamma_r,1,m)}   \end{align}
  When $\gamma_c, \gamma_r>1$, both the nominator and the dominator of $ \sum_{j=2}^{m}f_j p_j$ are $\Theta(1)$, from which (iv) follows.
 
Since the proof of the part (v)
 is similar to parts (i) and (ii), we omit it. 
  
 
\end{proof}

\begin{lemma} \label{E[v]scaling}
If $\gamma_c>1$, $\gamma_r<1$, $k=\Theta(m^{\eta_1})$, and $h=\Theta(1)$
 \begin{align}
\nonumber  E_v[v|A_{q,h}^{k}]=\Theta(\frac{1}{m^{\eta_1}})
 \end{align}
 where $\eta_1=\frac{\gamma_c(1-\gamma_r)}{1-\gamma_r+\gamma_c}$, $q=m^{\frac{\eta_1}{\gamma_c}}$, and $E_v[v|A_{q,h}^{k}]$ is defined in (\ref{E_v[v|A_{q,h}^{k}]}).
\end{lemma}
\begin{proof}
For the lower-bound, we have:
\begin{align}
\nonumber E_v[v|A_{q,h}^{k}]=&f_q+\sum_{j=q}^{m} f_j (1-(1-p_j)^{k-h})\\
 \geq & \sum_{j=q}^{m} f_j (1-e^{-k^{\prime} p_j})
\end{align}
where $k^{\prime}=k-h=\Theta(m^{\eta_1})$. Using the taylor series, we obtain:
\begin{align}
\nonumber &E_v[v|A_{q,h}^{k}]\geq  \sum_{j=q}^{m} f_j k^{\prime}p_j+f_j\frac{1}{2!}(k^{\prime}p_j)^2+f_j\frac{1}{3!}(k^{\prime}p_j)^3+\ldots\\
\nonumber &=k^{\prime}\frac{H(\gamma_c+\gamma_r,q,m)}{H(\gamma_r,1,m)H(\gamma_c,1,m)}+\frac{1}{2!}{k^{\prime}}^2\frac{H(2\gamma_c+\gamma_r,q,m)}{H(\gamma_r,1,m)H(\gamma_c,1,m)^2}\\
&+\frac{1}{3!}{k^{\prime}}^3\frac{H(3\gamma_c+\gamma_r,q,m)}{H(\gamma_r,1,m)H(\gamma_c,1,m)^3}+\ldots
\end{align}
Parts (i) and (ii) of  lemma \ref{scaling} imply that all terms in the above equation are $\Theta(\frac{1}{m^{\eta_1}})$. 

For showing the upper bound, 
\begin{align}
\nonumber E_v[v|A_{q,h}^{k}]&=f_q+\sum_{j=q+1}^{m}f_j(1-(1-p_j)^{{k}})\\
&\leq f_q+k\sum_{j=q}^{m}f_jp_j \label{two sum}\\
&\leq \frac{1}{q^{\gamma_r}H(\gamma_r,1,m)}+k \frac{H(\gamma_c+\gamma_r, q,m)}{H(\gamma_r,1,m)H(\gamma_c,1,m)}
\end{align}
 If we apply  the results of lemma \ref{scaling}, 
 we can show that $E_v[v|A_{q,h}^{k}]$ is $O(\frac{1}{m^{\eta_1}})$.  


 
\end{proof}

\begin{lemma}\label{azuma}
For $t<E_v[v|A_{q,h}^{k}$],
\begin{align}\label{azuma1}
\Pr\big[{|v-E_v[v|A_{q,h}^{k}]|\leq t}\big]\geq  1-2\exp\big(-\frac{2t^2}{(k-h)(f_q)^2}\big) 
\end{align}
\end{lemma}

\begin {proof}
Function $v:\{1,2,\ldots,m\}^k \rightarrow R$ is equal to 
\[v(\omega_1,\omega_2,\ldots, \omega_k)=\sum_{i\in {\tilde{ \omega}\cap Q}} f_i\]
where $\omega_j$ is the file that user $j$ stores, $\tilde{ \omega}=\cup_{j=1}^{k}\omega_j$ and $Q=\{q, q+1,\ldots,m\}$.
$v$ is the sum of popularity of union of files stored by users when only files in set $Q$ are considered to be valuable. 
By replacing the $i$th coordinate $\omega_i$ by some other value  the value of $v$ can change at most by $f_q$, \textit{i.e.}, 
\[\sup_{\omega_1,\ldots, \omega_k, \hat\omega_i} |v(\omega_1,\ldots, \omega_k)-v(\omega_1,\ldots, \hat\omega_i, \omega_{i+1} \ldots, \omega_k)|\leq f_q\]
Using Azuma-Hoefding inequality \cite{mitzenmacher2005probability},
\begin{align}\label{azuma1}
\Pr\big[{|v-E_v[v|A_{q,h}^{k}]|\geq t}\big]\leq 2\exp\big(-\frac{2t^2}{(k-h)(f_q)^2}\big) 
\end{align}
 the result follows.

\end{proof}

\bibliographystyle{IEEEtran}
\bibliography{ref_gupta.bib}

\begin{thebibliography}{10}
\providecommand{\url}[1]{#1}
\csname url@samestyle\endcsname
\providecommand{\newblock}{\relax}
\providecommand{\bibinfo}[2]{#2}
\providecommand{\BIBentrySTDinterwordspacing}{\spaceskip=0pt\relax}
\providecommand{\BIBentryALTinterwordstretchfactor}{4}
\providecommand{\BIBentryALTinterwordspacing}{\spaceskip=\fontdimen2\font plus
\BIBentryALTinterwordstretchfactor\fontdimen3\font minus
  \fontdimen4\font\relax}
\providecommand{\BIBforeignlanguage}[2]{{%
\expandafter\ifx\csname l@#1\endcsname\relax
\typeout{** WARNING: IEEEtran.bst: No hyphenation pattern has been}%
\typeout{** loaded for the language `#1'. Using the pattern for}%
\typeout{** the default language instead.}%
\else
\language=\csname l@#1\endcsname
\fi
#2}}
\providecommand{\BIBdecl}{\relax}
\BIBdecl

\bibitem{cisco66}
``http://www.cisco.com/en/us/solutions/collateral/ns341/ns525/ns537
  /ns705/ns827/white\_paper\_c11-520862.html.''

\bibitem{chandrasekhar2008femtocell}
V.~Chandrasekhar, J.~Andrews, and A.~Gatherer, ``Femtocell networks: a
  survey,'' \emph{Communications Magazine, IEEE}, vol.~46, no.~9, pp. 59--67,
  2008.

\bibitem{femtocaching}
N.~Golrezaei, K.~Shanmugam, A.~Dimakis, A.~Molisch, and G.~Caire,
  ``Femtocaching: Wireless video content delivery through distributed caching
  helpers,'' in \emph{INFOCOM}.\hskip 1em plus 0.5em minus 0.4em\relax IEEE,
  2012.

\bibitem{coded_femtocaching}
------, ``Wireless video content delivery through coded distributed caching,''
  in \emph{ICC}.\hskip 1em plus 0.5em minus 0.4em\relax IEEE, 2012.

\bibitem{ICC_Workshop}
N.~Golrezaei, A.~Molisch, and A.~Dimakis, ``Base-station assisted
  device-to-device communications for high-throughput wireless video
  networks,'' \emph{Accepted in ICC'12 WS - ViOpt}.

\bibitem{magazine}
N.~Golrezaei, A.~Molisch, A.~Dimakis, and G.~Caire, ``Femtocaching and
  device-to-device collaboration: A new architecture for wireless video
  distribution,'' \emph{Accepted in IEEE Communications Magazine}, 2012.

\bibitem{TWC}
D.~A. Golrezaei, N. and A.~Molisch, ``Device to device transmission for
  increasing video throughput in wireless networks,'' \emph{To be submitted for
  publication}.

\bibitem{gupta2000capacity}
P.~Gupta and P.~Kumar, ``The capacity of wireless networks,'' \emph{Information
  Theory, IEEE Transactions on}, vol.~46, no.~2, pp. 388--404, 2000.

\bibitem{Tse}
A.~Ozgur, O.~L{\'e}v{\^e}que, and D.~Tse, ``Hierarchical cooperation achieves
  linear capacity scaling in ad hoc networks,'' in \emph{INFOCOM 2007. 26th
  IEEE International Conference on Computer Communications. IEEE}.\hskip 1em
  plus 0.5em minus 0.4em\relax IEEE, 2007, pp. 382--390.

\bibitem{grossglauser2001mobility}
M.~Grossglauser and D.~Tse, ``Mobility increases the capacity of ad-hoc
  wireless networks,'' in \emph{INFOCOM 2001. Twentieth Annual Joint Conference
  of the IEEE Computer and Communications Societies. Proceedings. IEEE},
  vol.~3.\hskip 1em plus 0.5em minus 0.4em\relax IEEE, 2001, pp. 1360--1369.

\bibitem{franceschetti2009capacity}
M.~Franceschetti, M.~Migliore, and P.~Minero, ``The capacity of wireless
  networks: information-theoretic and physical limits,'' \emph{Information
  Theory, IEEE Transactions on}, vol.~55, no.~8, pp. 3413--3424, 2009.

\bibitem{tracedata}
``http://traces.cs.umass.edu/index.php/network/network.''

\bibitem{zipf}
M.~Cha, H.~Kwak, P.~Rodriguez, Y.~Ahn, and S.~Moon, ``I tube, you tube,
  everybody tubes: analyzing the world's largest user generated content video
  system,'' in \emph{Proceedings of the 7th ACM SIGCOMM conference on Internet
  measurement}.\hskip 1em plus 0.5em minus 0.4em\relax ACM, 2007, pp. 1--14.

\bibitem{RGGbook}
M.~Penrose and O.~U. Press, \emph{Random geometric graphs}.\hskip 1em plus
  0.5em minus 0.4em\relax Oxford University Press Oxford, 2003, vol.~5.

\bibitem{Molisch_book_2011}
A.~Molisch, \emph{Wireless communications}.\hskip 1em plus 0.5em minus
  0.4em\relax Wiley, 2011.

\bibitem{lawler1980generating}
E.~Lawler, J.~Lenstra, A.~Kan, and E.~U.~E. Institute, ``Generating all maximal
  independent sets: Np-hardness and polynomial-time algorithms,'' \emph{SIAM J.
  Comput.}, vol.~9, no.~3, pp. 558--565, 1980.

\bibitem{chen2011file}
Y.~Chen, C.~Caramanis, and S.~Shakkottai, ``On file sharing over a wireless
  social network,'' in \emph{Information Theory Proceedings (ISIT), 2011 IEEE
  International Symposium on}.\hskip 1em plus 0.5em minus 0.4em\relax IEEE,
  2011, pp. 249--253.

\bibitem{chernoff1952measure}
H.~Chernoff, ``A measure of asymptotic efficiency for tests of a hypothesis
  based on the sum of observations,'' \emph{The Annals of Mathematical
  Statistics}, vol.~23, no.~4, pp. 493--507, 1952.

\bibitem{FKG}
R.~Holley, ``Remarks on the fkg inequalities,'' \emph{Communications in
  Mathematical Physics}, vol.~36, no.~3, pp. 227--231, 1974.

\bibitem{mitzenmacher2005probability}
M.~Mitzenmacher and E.~Upfal, \emph{Probability and computing: Randomized
  algorithms and probabilistic analysis}.\hskip 1em plus 0.5em minus
  0.4em\relax Cambridge Univ Pr, 2005.

\end{thebibliography}

\end{document}